\documentclass[12pt]{article}
\usepackage{amsthm,amsmath,indentfirst,amssymb}
\usepackage{algorithm,algpseudocode}
\newtheorem{lemma}{Lemma}[section]
\newtheorem{theorem}[lemma]{Theorem}

\newtheorem{definition}[lemma]{Definition}

\begin{document}

\title{Approximation Algorithm for Minimum Weight Connected $m$-Fold Dominating Set}
\author{\footnotesize Zhao Zhang$^{1}$\quad Jiao Zhou$^{2}$\quad Ker-I Ko$^3$\quad Ding-zhu Du$^4$\\
    {\it\small $^1$ College of Mathematics Physics and Information Engineering, Zhejiang Normal University}\\
    {\it\small Jinhua, Zhejiang, 321004, China}\\
   {\it\small $^2$ College of Mathematics and System Sciences, Xinjiang University}\\
    {\it\small Urumqi, Xinjiang, 830046, China}\\
    {\it\small $^3$ Department of Computer Science, National Chiao Tung University}\\
    {\it\small Hsinchu, 30050, Taiwan}\\
    {\it\small $^4$ Department of Computer Science, University of Texas at Dallas}\\
    {\it\small Richardson, Texas, 75080, USA}}
\date{}
\maketitle
{\bf Abstract}: Using connected dominating set (CDS) to serve as a virtual backbone in a wireless networks can save energy and reduce interference. Since nodes may fail due to accidental damage or energy depletion, it is desirable that the virtual backbone has some fault-tolerance. A $k$-connected $m$-fold dominating set ($(k,m)$-CDS) of a graph $G$ is a node set $D$ such that every node in $V\setminus D$ has at least $m$ neighbors in $D$ and the subgraph of $G$ induced by $D$ is $k$-connected. Using $(k,m)$-CDS can tolerate the failure of $\min\{k-1,m-1\}$ nodes. In this paper, we study Minimum Weight $(1,m)$-CDS problem ($(1,m)$-MWCDS), and present an $(H(\delta+m)+2H(\delta-1))$-approximation algorithm, where $\delta$ is the maximum degree of the graph and $H(\cdot)$ is the Harmonic number. Notice that there is a $1.35\ln n$-approximation algorithm for the $(1,1)$-MWCDS problem, where $n$ is the number of nodes in the graph. Though our constant in $O(\ln \cdot)$ is larger than 1.35, $n$ is replaced by $\delta$. Such a replacement enables us to obtain a $3.67$-approximation for the connecting part of $(1,m)$-MWCDS problem on unit disk graphs.


\section{Introduction}
A wireless sensor network (WSN) consists of spatially distributed autonomous sensors to monitor physical or environmental condition, and to cooperatively pass their data through the network. During recent years, WSN has been widely used in many fields, such as environment and habitat monitoring, disaster recovery, health applications, etc. Since there is no fixed or predefined infrastructure in WSNs, frequent flooding of control messages from sensors may cause a lot of redundant contentions and collisions. Therefore, people have proposed the concept of virtual backbone which corresponds to a connected dominating set in a graph (Das and Bhargharan \cite{Das} and Ephremides {\it et al.} \cite{Ephremides}).

Given a graph $G$ with node set $V$ and edge set $E$, a subset of nodes $C\subseteq V$ is said to be a {\em dominating set} (DS) of $G$ if any $v\in V\setminus C$ is adjacent to at least one node of $C$. We say that a dominating set $C$ of $G$ is a {\em connected dominating set} of $G$ if the subgraph of $G$ induced by $C$, denoted by $G[C]$, is connected. Nodes in $C$ are called {\em dominators}, the nodes in $V\setminus C$ are called {\em dominatees}.

Because sensors in a WSN are prone to failures due to accidental damage or battery depletion, it is important to maintain a certain degree of redundancy such that the virtual backbone is more fault-tolerant. In a more general setting, every sensor has a cost, it is desirable that under the condition that tasks can be successfully accomplished, the whole cost of virtual backbone is as small as possible. These considerations lead to the Minimum Node-Weighted $k$-Connected $m$-Fold Dominating Set problem (abbreviated as $(k,m)$-MWCDS), which is defined as follows:

\begin{definition}[$(k,m)$-MWCDS]
{\rm Let $G$ be a connected graph, $k$ and $m$ be two positive integers, $c:V\rightarrow R^{+}$ be a cost function on nodes. A node subset $D\subseteq V$ is an $m$-fold dominating set ($m$-DS) if every node in $V\setminus D$ has at least $m$ neighbors in $D$. It is a $k$-connected $m$-fold dominating set ($(k,m)$-CDS) if furthermore, the subgraph of $G$ induced by $D$ is $k$-connected. The $(k,m)$-MWCDS problem is to find a $(k,m)$-CDS $D$ such that the cost of $D$ is minimized, that is, $c(D)=\sum_{u\in D}c(u)$ is as small as possible.}
\end{definition}

After Dai and Wu \cite{Dai} proposed using $(k,k)$-CDS as a model for fault-tolerant virtual backbone, a lot of approximation algorithms emerged, most of which are on unit disk graphs. In a {\rm unit disk graph} (UDG), every node corresponds to a sensor on the plane, two nodes are adjacent if and only if the Euclidean distance between the corresponding sensors is at most one unit. There are a lot of studies on fault-tolerant virtual backbone in UDG which assume unit weight on each disk. However, for a general graph with a general weight function, related studies are rare.

In this paper, we present a $(H(\delta+m)+2H(\delta-1))$-approximation algorithm for the $(1,m)$-MWCDS problem on a general graph, where $\delta$ is the maximum degree of the graph, and $H(\gamma)=\sum_{1}^\gamma 1/i$ is the Harmonic number. It is a two-phase greedy algorithm. First, it constructs an $m$-fold dominating set $D_1$ of $G$. Then it connects $D_1$ by adding a set of connectors $D_2$. It is well known that if the potential function and the cost function related with a greedy algorithm is monotone increasing and submodular, then an $O(\ln n)$ performance ration can be achieved. Unfortunately, for various minimum CDS problems, no such potential functions are known. Nevertheless, we manage, in this paper, to deal with a nonsubmodular potential function and achieve an $O(\ln \delta)$ performance ratio.

It should be pointed out that for a general graph, Guha and Khuller \cite{Guha} proposed a $1.35\ln n$-approximation for the $(1,1)$-MWCDS problem. Though our constant in $O(\ln n)$ is larger, the parameter $n$ in the performance ratio is replaced by $\delta$. In many cases, $\delta$ might be substantially smaller than $n$. In particular, for an UDG, due to such an replacement, after having found an $m$-DS, the connecting part has a performance ratio at most $3.67$. In \cite{Zou1}, Zou {\it et al.} proposed a method for the connecting part with performance ratio at most $3.875$, which makes use of a $1.55$-approximation algorithm \cite{Robin} for the classic Minimum Steiner Tree problem. If using currently best ratio for Minimum Steiner Tree problem \cite{Byrka}, then their ratio is at most $3.475$. Notice that the algorithm in \cite{Byrka} uses randomized iterative rounding. So, although our ratio 3.67 is a litter larger than $3.475$, our algorithm has the advantage that it is purely combinatorial. Furthermore, we believe that our method is also of more theoretical interests and may find more applications in the study of other related problems.

The rest of this paper is organized as follows. Section \ref{sec.2} introduces related works. Some notation and some preliminary are given in Section \ref{sec.3}. In Section \ref{sec-algorithm}, the algorithm is presented. Section \ref{sec.5} analyzes the performance ratio. Section \ref{secUDG} improves the ratio on unit disk graph. Section \ref{sec.7} concludes the paper.

\section{Related work}\label{sec.2}

The idea of using a connected dominating set as a virtual backbone for WSN was proposed by Das and Bhargharan \cite{Das} and Ephremides {\it et al.} \cite{Ephremides}. Constructing a CDS of the minimum size is NP-hard. In fact, Guha and Khuller \cite{Guha} proved that a minimum CDS cannot be approximated within $\rho\ln n$ for any $0<\rho<1$ unless $NP\subseteq DTIME(n^{O(loglog n)})$. In the same paper, they proposed two greedy algorithms with performance ratios of $2(H(\delta)+1)$ and $H(\delta)+2$, respectively, where $\delta$ is the maximum degree of the graph and $H(\cdot)$ is the harmonic number. This was improved by Ruan {\it et al.} \cite{Ruan} to $2+\ln\delta$. Du {\it et al.} \cite{Du} presented a $(1+\varepsilon)\ln(\delta-1)$-approximation algorithm, where $\varepsilon$ is an arbitrary positive real number.

For unit disk graphs, a polynomial time approximation scheme (PTAS) was given by Cheng {\it et al.} \cite{Cheng}, which was generalized to higher dimensional space by Zhang {\it et al.} \cite{Zhang1}. There are a lot of studies on distributed algorithms for this problem. For a comprehensive study on CDS in UDG, the readers may refer to the book \cite{DuBookCDS}.

Considering the weighted version of the CDS problem, Guha and Khuller \cite{Guha} proposed a $(c_{n}+1)\ln n$-approximation algorithm in a general graph, where $c_{n}\ln k$ is the performance ratio for the node weighted Steiner tree problem ($k$ is the number of terminal nodes to be connected). Later, they \cite{Guha1} improved it to an algorithm of performance ratio at most $(1.35+\varepsilon)\ln n$. For the minimum weight CDS in UDG, Zou {\it et al.} \cite{Zou1} gave a $(9.875+\varepsilon)$-approximation.

The problem of constructing fault-tolerant virtual backbones was introduced by Dai and Wu \cite{Dai}. They proposed three heuristic algorithms for the minimum $(k,k)$-CDS problem. However, no performance ratio analysis was given. A lot of works have been done for the CDS problem in UDG. The first constant approximation algorithm in this aspect was given by Wang {\it et al.} \cite{Wang}, who obtained a 72-approximation for the $(2,1)$-CDS problem in UDG. Shang {\it et al.} \cite{Shang} gave an algorithm for the minimum $(1,m)$-CDS problem and an algorithm for the minimum $(2,m)$-CDS problem in UDG, the performance ratios are $5+\frac{5}{m}$ for $m\leq5$ and $7$ for $m>5$, and $5+\frac{25}{m}$ for $2\leq m\leq5$ and $11$ for $m>5$, respectively. Constant approximation algorithms also exist for $(3,m)$-CDS in UDG \cite{WangWei,WangWei2}. Recently, Shi {\it et al.} \cite{ShiZhang} presented the first constant approximation algorithm for general $(k,m)$-CDS on UDG.

For the fault-tolerant CDS problem in a general graph, Zhang {\it et al.} \cite{Zhang2} gave a $2rH(\delta_r+m-1)$-approximation
for the minimum $r$-hop $(1,m)$-CDS problem, where $\delta_r$ is the maximum degree in $G^r$, the $r$-th power graph of $G$. A node $u$ is $r$-hop dominated by a set $D$ if it is at most $r$-hops away from $D$. In particular, taking $r=1$, the algorithm in \cite{Zhang2} has performance ratio at most $2H(\delta+m-1)$ for the minimum $(1,m)$-CDS problem. This was improved by our recent work \cite{Zhou} to $2+H(\delta+m-2)$. We also gave an $(\ln\delta+o(\ln\delta))$-approximation algorithm for the minimum $(2,m)$-CDS problem \cite{Shi} and the minimum $(3,m)$-CDS problem \cite{ZhangZhou} on a general graph.

For the weighted version of fault-tolerant CDS problem on UDG, as a consequence of recent work \cite{LiJian}, the minimum weight $(1,1)$-CDS problem admits a PTAS. Combining the constant approximation algorithm for the minimum weight $m$-fold dominating set problem \cite{Fukunage} and the 3.475-approximation algorithm for the connecting part \cite{Zou1}, $(1,m)$-MWCDS on UDG admits a constant-approximation. Recently, we \cite{ShiZhang,ZhangShi} gave the first constant approximation algorithm for the general minimum weight $(k,m)$-CDS problem on unit disk graph.

As far as we know, there is no previous work on the approximation of the weighted version of fault-tolerant CDS problem in a general graph.

Notice that in \cite{Zou1}, the $3.875$-approximation for the connecting part is based on the 1.55-approximation algorithm \cite{Robin} for the classic minimum Steiner tree problem. If the best known ratio for the Steiner tree problem is used, which is $1.39$ currently, then their connecting part has performance ratio at most $3.475$. Notice that the $1.39$-approximation for the Steiner tree problem uses randomized iterative rounding. So, although our performance ratio is larger than $3.475$, it has the advantage that it is purely combinatorial.

\section{Preliminaries}\label{sec.3}

In this section, we introduce some notation and give some preliminary results. For a node $u\in V(G)$, denote by $N_G(u)$ the set of neighbors of $u$ in $G$, and $deg_G(u)=|N_G(u)|$ is the degree of node $u$ in $G$. For a node subset $D\subseteq V(G)$, $N_G(D)=\bigcup_{u\in D}N_G(u)\setminus D$ is the neighbor set of $D$, $G[D]$ is the subgraph of $G$ induced by $D$. When there is no confusion in the context, the vertex set of a subgraph will be used to denote the subgraph itself.

For an element set $U$, suppose $f:2^U\mapsto \mathbb R^+$ is a set function on $U$ (called a {\em potential function}). For two element sets $C,D\subseteq U$, let
$$
\bigtriangleup_Df(C)=f(C\cup D)-f(C)
$$
be the marginal profit obtained by adding $D$ into $C$. For simplicity, $\bigtriangleup_uf(C)$ will be used to denote $\bigtriangleup_{\{u\}}f(C)$ when $u$ is a node. Potential function $f$ is {\em monotone increasing} if $f(C)\leq f(D)$ holds for any subsets $C\subseteq D\subseteq V$. It is {\em submodular} if and only if $\bigtriangleup_uf(C)\geq\bigtriangleup_uf(D)$ holds for any $C\subseteq D\subseteq V$ and any $u\in V\setminus D$. A monotone increasing and submodular function $f$ with $f(\emptyset)=0$ is called a {\em polymatroid}. Given an element set $U$ with cost function $c:U\mapsto \mathbb R^+$ and given a polymatroid $f:2^U\mapsto \mathbb R^+$, denote by $\Omega_f=\{C\subseteq U\colon \bigtriangleup_uf(C)=0\ \mbox{for any $u\in U$}\}$. The {\em Submodular Cover} problem is:
\begin{align*}
\min &\ \ c(C)=\sum_{u\in C}c(u)\\
s.t. &\ \ C\in\Omega_f.
\end{align*}

The following is a classic result which can be found in \cite{DuBook} Theorem 2.29.
\begin{theorem}\label{thm-submodular-cover}
The greedy algorithm for the submodular cover problem has performance ratio $H(\gamma)$, where $\gamma=\max\{f(\{u\})\colon u\in U\}$ and $H(\gamma)=\sum_{i=1}^\gamma 1/i$ is the Harmonic number.
\end{theorem}

\section{The Algorithm}\label{sec-algorithm}

\subsection{The Algorithm}

For a node set $C$, denote by $p(C)$ the number of components in $G[C]$. For a node $u\in V\setminus C$, let $N_C(u)$ denote the set of neighbors of $u$ in $C$. Define
$$
q_C(u)=\left\{\begin{array}{ll}\max\{m-|N_C(u)|,0\}, & u\in V\setminus C,\\ 0, & u\in C,\end{array}\right.
$$
and
$$
q(C)=m|V|-\sum_{u\in V}q_C(u).
$$

For a node set $U\subseteq V$, denote by $NC_C(U)$ the set of components of $G[C]$ which are adjacent to $U$. Every component in $NC_C(U)$ is called a {\em component neighbor} of $U$ in $G[C]$. For a node $u\in V$, we shall use $S_u$ to denote some star with center $u$, that is, $S_u$ is a subgraph of $G$ induced by edges between node $u$ and some of $u$'s neighbors in $G$. In particular, a node is a star of cardinality one and  an edge is a star of cardinality two. To abuse the notation a little, we also use $S_u$ to denote the set of nodes in $S_u$. Suppose $S_u\setminus\{u\}=\{u_{1},u_{2},\ldots,u_{s}\}$, where $c(u_{1})\leq c(u_{2})\leq\ldots\leq c(u_{s})$. Define
\begin{equation}\label{eq14-11-14-2}
p'_{C}(S_u)=|NC_{C}(u)|-1+\sum_{i=1}^{s}\min\{1,
-\bigtriangleup_{u_{i}}p(C\cup\{u,u_{1},\ldots,u_{i-1}\})\}.
\end{equation}
Call $e_C(S_u)=p'_C(S_u)/c(S_u)$ the {\em efficiency of $S_u$ with respect to $C$}.

The algorithm is presented in Algorithm \ref{alg10-6-1}.

\begin{algorithm}[h!]
\caption{\textbf{}}
 Input: A connected graph $G=(V,E)$.

Output: A $(1,m)$-CDS $D_G$ of $G$.
\begin{algorithmic}[1]
    \State Set $D_1$ $\leftarrow$ $\emptyset$
    \While{there exists a node $u\in V\setminus D_1$ such that $\bigtriangleup_uq(D_1)> 0$, }
       \State select $u$ which maximizes $\bigtriangleup_u q(D_1)/c(u)$
       \State $D_1 \leftarrow D_1\cup \{u\}$
    \EndWhile
      \State Set $D_2\leftarrow\emptyset$
    \While{there exists a star $S_u\subseteq V\setminus (D_1\cup D_2)$ such that $p_{D_1\cup D_2}'(S_u)> 0$, }
       \State select a star $S_u$ with the largest efficiency with respect to $D_1\cup D_2$. 
       \State $D_2 \leftarrow D_2\cup S_u$
    \EndWhile
    \State Output  $D_{G}\leftarrow D_1\cup D_2$
\end{algorithmic}\label{alg10-6-1}
\end{algorithm}

\subsection{The Idea of the Algorithm}
The idea underlying the algorithm is as follows. Potential function $q_{D_1}(u)$ measures how many more times that node $u$ needs to be dominated, and $q(D_1)$ is the total residual domination requirement. As can be seen from Lemma \ref{le10-8-1}, at the end of the first phase, we have an $m$-fold dominating set $D_1$. Then, the second phase aims to connect it by adding a connector set $D_2$. A natural potential function for connection is $p(D_1\cup D_2)$, the number of components of $G[D_1\cup D_2]$. That is, every iteration chooses a node set $S$ to be added into $D_2$ which reduces the number of components by the largest amount until $p(D_1\cup D_2)$ reaches 1. It is a folklore result (see Lemma \ref{lem14-11-22-6}) that simultaneously adding at most two nodes can reduce the number of components (even when ading any single node does not reduce the number of components). So, it is natural to use
\begin{equation}\label{eq14-11-14-1}
\max\{-\bigtriangleup_Sp(D_1\cup D_2)/c(S)\colon S\subseteq V\setminus (D_1\cup D_2),1\leq |S|\leq 2\}
\end{equation}
to work as a criterion for the choice of node set $S$ to be added into $D_2$. However, choosing at most two nodes might yield a solution with very bad performance ratio. Consider the example shown in Fig.\ref{fig11-01-1}, its optimal solution $OPT=\{u,v_1,\ldots,v_d\}$ has cost $opt=1+(d+1)\varepsilon$. If \eqref{eq14-11-14-1} is used as the greedy criterion, then the output is $\bigcup_{i=1}^{d}\{u_{i},v_{i}\}$, whose cost is $c(\bigcup_{i=1}^{d}\{u_{i},v_{i}\})=d(1+\varepsilon)\approx d\cdot opt=\frac{n-2}{3}opt$. To overcome such a shortcoming, an idea is to choose some star $S_u$ to maximize $-\bigtriangleup_{S_u}p(D_1\cup D_2)/c(S_u)$ in each iteration. However, the computation of a most efficient star will take exponential time. This is why we define $p'$ as in \eqref{eq14-11-14-2} to be used in the greedy criterion. On one hand, $p'$ also plays the role of counting the reduction on components. On the other hand, a most efficient $S_u$ to maximize $p'_{D_1\cup D_2}(S_u)/c(S_u)$ can be found in polynomial time, which will be shown in Subsection \ref{subsec4.2}.

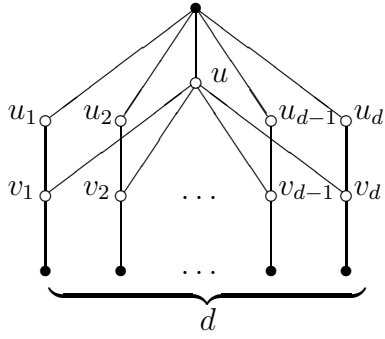
\begin{figure*}[!htbp]
\setlength{\unitlength}{1mm}
\begin{center}
\begin{picture}(100,70)
\put(50,65){\circle*{1.5}}
\put(50,55){\circle{1.5}}
\put(40,50){\circle{1.5}}
\put(30,50){\circle{1.5}}
\put(60,50){\circle{1.5}}
\put(70,50){\circle{1.5}}

\put(50,65){\line(0,-1){9.25}}
\put(50,65){\line(-2,-3){9.55}}
\put(50,65){\line(-4,-3){19.35}}
\put(50,65){\line(2,-3){9.55}}
\put(50,65){\line(4,-3){19.35}}

\put(40,40){\circle{1.5}}
\put(30,40){\circle{1.5}}
\put(60,40){\circle{1.5}}
\put(70,40){\circle{1.5}}
\put(40,40.75){\line(0,1){8.5}}
\put(30,40.75){\line(0,1){8.5}}
\put(60,40.75){\line(0,1){8.5}}
\put(70,40.75){\line(0,1){8.5}}

\put(30.5,40.5){\line(4,3){18.9}}
\put(40.5,40.5){\line(2,3){9.25}}
\put(59.5,40.5){\line(-2,3){9.25}}
\put(69.5,40.5){\line(-4,3){18.9}}
\put(48,28){\makebox(4,2)[tl]{\hbox{$\ldots$}}}
\put(48,38){\makebox(4,2)[tl]{\hbox{$\ldots$}}}
\put(30.5,28){$\underbrace{\makebox{}\qquad\qquad\qquad\qquad\qquad\,}$}
\put(50.5,23){\makebox(4,2)[tl]{\hbox{$d$}}}

\put(40,30){\circle*{1.5}}
\put(30,30){\circle*{1.5}}
\put(60,30){\circle*{1.5}}
\put(70,30){\circle*{1.5}}
\put(30,39.25){\line(0,-1){9}}
\put(40,39.25){\line(0,-1){9}}
\put(60,39.25){\line(0,-1){9}}
\put(70,39.25){\line(0,-1){9}}

\put(52,55){\makebox(4,2)[tl]{\hbox{$u$}}}
\put(25,50){\makebox(4,2)[tl]{\hbox{$u_{1}$}}}
\put(35,50){\makebox(4,2)[tl]{\hbox{$u_{2}$}}}
\put(61,50){\makebox(4,2)[tl]{\hbox{$u_{d-1}$}}}
\put(71,50){\makebox(4,2)[tl]{\hbox{$u_{d}$}}}
\put(25,40){\makebox(4,2)[tl]{\hbox{$v_{1}$}}}
\put(35,40){\makebox(4,2)[tl]{\hbox{$v_{2}$}}}

\put(61,40){\makebox(4,2)[tl]{\hbox{$v_{d-1}$}}}
\put(71,40){\makebox(4,2)[tl]{\hbox{$v_{d}$}}}

\end{picture}
\vskip -2.5cm\caption{Solid nodes represent nodes in $C$. The costs on circled nodes are $c(u_{1})=c(u_{2})=\ldots=c(u_{d})=1$,
$c(v_{1})=c(v_{2})=\ldots=c(v_{d})=\varepsilon$, and $c(u)=1+\varepsilon$.}\label{fig11-01-1}
\end{center}
\end{figure*}

\section{The Analysis of the Algorithm}\label{sec.5}

In this section, we analyze the performance ratio of Algorithm \ref{alg10-6-1}.

\subsection{The Analysis of $D_1$}

\begin{lemma}\label{lem14-11-14-3}
Function $q$ is a polymatroid.
\end{lemma}
\begin{proof}
Obviously, $q(\emptyset)=0$. It is easy to see that for any node $u\in V$, function $-q_C(u)$ is monotone increasing and submodular with respect to $C$. So, function $q$, being the summation of a constant function and some monotone increasing and submodular functions, is also monotone increasing and submodular.
\end{proof}

\begin{lemma}\label{le10-8-1}
The final node set $D_1$ in Algorithm \ref{alg10-6-1} is an $m$-fold dominating set of $G$.
\end{lemma}
\begin{proof}
If there exists a node $u\in V\backslash D_1$ with $|N_{D_1}(u)|<m$, then $-\bigtriangleup_uq_{D_1}(u)=-q_{D_1\cup\{u\}}(u)+q_{D_1}(u)=-0+(m-|N_{D_1}(u)|)>0$, and $-\bigtriangleup_uq_{D_1}(v)\geq 0$ for any node $v\in V\setminus \{u\}$ (by the monotonicity of $-q_{D_1}(v)$ with respect to $D_1$). Thus $\bigtriangleup_{u}q(D_1)=-\sum_{v\in V}\bigtriangleup_{u}q_{D_1}(v)>0$, and the algorithm does not terminate at this stage.
\end{proof}

\begin{theorem}\label{thm14-11-25-1}
The final node set $D_1$ in Algorithm \ref{alg10-6-1} has weight $w(D_1)\leq H(m+\delta)\cdot opt$, where $\delta$ is the maximum degree of graph $G$, and $opt$ is the optimal value for the $m$-MWCDS problem.
\end{theorem}
\begin{proof}
By Lemma \ref{lem14-11-14-3} and Lemma \ref{le10-8-1}, the minimum weight $m$-fold dominating set problem is a special submodular cover problem with potential function $q$. So by Theorem \ref{thm-submodular-cover}, we have $w(D_1)\leq H(\gamma)\cdot opt'$, where $opt'$ is the optimal value for the $m$-MWDS problem. Then, the result follows from the observation that $opt'\leq opt$ and $\gamma=\max\limits_{u\in V}\{q(\{u\})\}=\max\limits_{u\in V}\{m|V|-\sum_{v\in V\setminus\{u\}}(m-|N_{\{u\}}(v)|)\}=\max\limits_{u\in V}\{m+d_G(u)\}=m+\delta$.
\end{proof}

\subsection{The Computation of an Optimal Star for Greedy Choice}\label{subsec4.2}

The idea for the definition of $p'_C(S_u)$ is as follows: Adding node $u$ into $C$ will merge those components in $NC_C(u)$ into one component of $G[C\cup \{u\}]$, say $\widetilde C$, which reduces the number of components by $|NC_C(u)|-1$. Then adding nodes in $S_u\setminus \{u\}$ sequentially according to the increasing order of their costs. Notice that \begin{equation}\label{eq14-11-22-5}
-\bigtriangleup_{u_i}p(C\cup\{u_1,\ldots,u_{i-1}\})=|NC_C(u_i)\setminus NC_C(u,u_1,\ldots,u_{i-1})|
\end{equation}
is the number of components newly merged into $\widetilde C$. So, the term in the summation of definition \eqref{eq14-11-14-2} indicates that if adding $u_i$ merges at least one more component, we regard its contribution to $p'_C(S_u)$ as one. The advantage of such a counting is that an optimal star for the greedy criterion is polynomial-time computable. This claim is based on the following lemma. A simple relation will be used in the proof: for four positive real numbers $a,b,c,d$
\begin{equation}\label{eq14-11-16-4}
\frac{a+b}{c+d}\geq \frac{b}{d}\Rightarrow \frac{a}{c}\geq \frac{b}{d}.
\end{equation}

\begin{lemma}\label{le10-30-1}
Suppose $C$ is a dominating set of graph $G$. Then, there exists a most efficient star $S_u$ with respect to $C$ such that $|NC_C(v)|=1$ for every node $v\in S_u\setminus \{u\}$. Furthermore, if we denote by $C_v$ the unique component in $NC_C(v)$, then components in $\{C_v\}_{v\in S_u\setminus\{u\}}$ are all distinct and they are also distinct from those components in $NC_C(u)$.
\end{lemma}
\begin{proof}
Suppose $S_u$ is a most efficient star with $S_u\setminus \{u\}=\{u_1,\ldots,u_s\}$ such that $c(u_1)\leq \ldots\leq c(u_s)$. We first show that for any $i=1,\ldots,s$,
\begin{equation}\label{eq14-11-22-2}
-\bigtriangleup_{u_{i}}p(C\cup\{u,u_{1},\ldots,u_{i-1}\})\geq 1.
\end{equation}
Suppose this is not true. Let $i$ be the first index with $-\bigtriangleup_{u_{i}}p(C\cup\{u,u_{1},\ldots,u_{i-1}\})=0$. Then by \eqref{eq14-11-22-5}, we have $NC_C(u_i)\subseteq NC_C(u,u_1,\ldots,u_{i-1})$ and thus $NC_C(u,u_1,\ldots,u_{i-1})=NC_C(u,u_1,\ldots,u_i)$. It follows that for any $j>i$,
$$
-\bigtriangleup_{u_{j}}p(C\cup\{u,u_{1},\ldots,u_{i-1},u_{i+1},\ldots,u_{j-1}\})=-\bigtriangleup_{u_{i}}p(C\cup\{u,u_{1},\ldots,u_{j-1}\}).
$$
So, for the star $S_u'=S_u\setminus\{u_i\}$, we have $p'_C(S'_u)=p'_C(S_u)$ and thus $p'_C(S'_u)/c(S'_u)>p'_C(S_u)/c(S_u)$, contradicting the maximality of $p_C'(S_u)/c(S_u)$.

As a consequence of \eqref{eq14-11-22-2},
$$
p_C'(S_u)=|NC_C(u)|-1+s.
$$

Suppose there is a node $u_i$ with $|NC_C(u_i)|\geq 2$, we choose $u_i$ to be such that $i$ is as small as possible. Let $S_u'=S_u\setminus \{u_i\}$. Notice that for any $j>i$,
$$
-\bigtriangleup_{u_j}p(C\cup\{u,u_{1},\ldots,u_{i-1},u_{i+1},\ldots,u_{j-1}\})\geq -\bigtriangleup_{u_{i}}p(C\cup\{u,u_{1},\ldots,u_{j-1}\}).
$$
Combining this with \eqref{eq14-11-22-2}, we have
$$
p_C'(S_u')=|NC_C(u)|-1+(s-1)=p_C'(S_u)-1.
$$
By the maximality of $S_u$, we have
$$
\frac{p_C'(S_u')}{c(S_u')}\leq \frac{p_C'(S_u)}{c(S_u)}=\frac{p_C'(S_u')+1}{c(S_u')+c(u_i)}.
$$
Then by \eqref{eq14-11-16-4},
$$
\frac{p_C'(S_u')}{c(S_u')}\leq \frac{1}{c(u_i)}.
$$
It follows that
$$
\frac{p_C'(S_u)}{c(S_u)}=\frac{p_C'(S_u')+1}{c(S_u')+c(u_i)}\leq \frac{\frac{1}{c(u_i)}c(S_u')+1}{c(S_u')+c(u_i)}=\frac{1}{c(u_i)}\leq \frac{p_C'(u_i)}{c(u_i)},
$$
where the last inequality holds because $p_C'(u_i)=|NC_C(u_i)|-1\geq 1$. Hence $u_i$ is a also a most efficient star. It is a trivial star satisfing the requirement of the lemma.

Next, suppose $S_u$ is a nontrivial star in which every $v\in S_u\setminus\{u\}$ has $|NC_C(v)|=1$. Notice that an equivalent statement of \eqref{eq14-11-22-2} is that $|NC_C(u_i)\setminus NC_C(u,u_1,\ldots,u_{i-1})|\geq 1$ for any $i=1,\ldots,s$. The second part of this lemma follows.
\end{proof}

By Lemma \ref{le10-30-1}, a most efficient star with respect to $C$ can be found in the following way. Guessing the center of the star requires time $O(n)$. Suppose $u$ is the guessed center. Let $N^{(u)}=\{v\colon v\ \mbox{is a neighbor of}\ u\ \mbox{in}\ G\ \mbox{and}\ |NC_C(v)|=1\}$. Order the nodes in $N^{(u)}$ as $u_1,\ldots,u_s$ such that $c(u_1)\leq \ldots\leq c(u_s)$. For $i=1,\ldots,s$, scan $u_i$ sequentially. If $u_i$ has $NC_C(u_i)\subseteq NC_C(u,u_1,\ldots,u_{i-1})$, then remove it from $N^{(u)}$. For convenience of statement, suppose the remaining set $N^{(u)}=\{u_1,u_2\ldots,u_t\}$. Then, the node set of a most efficient star centered at $u$ must be of the form $\{u,u_1,u_2,\ldots,u_l\}$ for some $l\in\{0,\ldots,t\}$. So, it suffices to compute the efficiency of the $t+1$ sets $\{u,u_1,\ldots,u_l\}$ for $l=0,\ldots,t$ and choose the most efficient one from them. Clearly, such a computation can be done in polynomial time.

\subsection{Correctness of the Algorithm}

The following result is a folklore in the study of CDS (see, for example, \cite{Wan}).

\begin{lemma}\label{lem14-11-22-6}
Suppose $D$ is a dominating set of $G$ such that $G[D]$ is not connected. Then, two nearest components of $G[D]$ are at most three hops away.
\end{lemma}

\begin{theorem}
The output $D_G$ of Algorithm \ref{alg10-6-1} is a $(1,m)$-CDS.
\end{theorem}
\begin{proof}
By Lemma \ref{le10-8-1}, $D_1$ is an $m$-fold DS, and thus $D_G$ is also an $m$-DS. If $G[D_G]$ is not connected, consider two nearest components of $G[D_G]$, say $G_1$ and $G_2$. Let $P=u_0u_1\ldots u_t$ be a shortest path between $G_1$ and $G_2$, where $u_0\in V(G_1)$ and $u_t\in V(G_2)$. By Lemma \ref{lem14-11-22-6}, we have $t=2$ or $3$. Then, $u_1$ (in the case $t=2$) or $u_1u_2$ (in the case $t=3$) is a star $S_{u_1}$ with $p_{D_G}'(S_{u_1})>0$. The algorithm will not terminate.
\end{proof}

\subsection{Decomposition of Optimal Solution}

The following lemma, as well as its proof, can be illustrated by Fig.\ref{fig11-1-2}.

\begin{lemma}\label{lem10-6-3}
Let $G$ be a connected graph, $C$ be a dominating set of $G$ and $C^{*}$ be a connected dominating set of $G$. Then $C^{*}\backslash C$ can be decomposed into the union of node sets $C^{*}\backslash C=Y_0\cup Y_{1}\cup Y_{2}\cup,\ldots,\cup Y_{h}$ such that:

$(\romannumeral1)$ for $1\leq i\leq h$, subgraph $G[Y_{i}]$ contains a star;

$(\romannumeral2)$ subgraph $G[C\cup Y_{1}\cup Y_{2}\cup,\ldots,\cup Y_{h}]$ is connected;

$(\romannumeral3)$ for $1\leq i\leq h$, $|NC_C(Y_i)|\geq 2$;

$(\romannumeral4)$ any node of $C^{*}\backslash C$ belongs to at most two sets of $\{Y_0,Y_{1},Y_{2},\ldots,Y_{h}\}$.
\end{lemma}
\begin{proof}
Let $H$ be the graph obtained from $G[C\cup C^{*}]$ by contracting every component of $G[C]$ into a super-node (call it a terminal node). Since $G[C\cup C^{*}]$ is connected, $H$ is also connected, and thus $H$ contains a spanning tree $T$. Recursively pruning non-terminal leaves, we obtain a tree $T'$ in which every leaf is a terminal node (see Fig.\ref{fig11-1-2}(a)(b)). Let $Y_0$ be the set of pruned nodes. We may assume that
\begin{equation}\label{eq14-11-23-2}
\mbox{every non-terminal node has at least one terminal neighbor in $T'$.}
\end{equation}
In fact, if this is not true, then we can modify $T'$ into another tree satisfying this assumption. See Fig.\ref{fig11-1-2}(c)(d) for an illustration. In $(c)$, node $u$ does not have a terminal neighbor in $T'$. Since $C$ is a dominating set, $u$ is adjacent with some component of $G[C]$, say the component corresponding to terminal node $v$. Adding edge $uv$ creates a unique cycle in $T'+uv$. Removing the edge on this cycle which is incident with $u$ in $T'$, we have another tree in which $u$ has a terminal neighbor (see Fig.\ref{fig11-1-2}(d)). Notice that the removed edge is between two non-terminal nodes. So, the number of non-terminal nodes which have no terminal neighbors is strictly reduced. Recursively making such a modification results in a tree satisfying assumption \eqref{eq14-11-23-2}.

Tree $T'$ can be viewed as a Steiner tree. By splitting off at non-leaf terminal nodes, $T'$ can be decomposed into full components $T_1',\ldots,T_l'$ (a {\em full component} in a Steiner tree is a subtree in which a node is a leaf if and only if it is a terminal node, see Fig.\ref{eq14-11-23-2}(e)). Let $T_i$ be the subtree of $T_i'$ induced by those non-terminal nodes.

For each $i\in\{1,\ldots,l\}$, let $v_{i}$ be an arbitrary node of $T_{i}$ and view $T_i$ as a tree rooted at $v_i$. For each node $v\in V(T_i)$, let $S_v$ be the star centered at node $v$ which contains all children of $v$ in $T_i$. Let $\mathcal S_i=\{S_v\colon v\in V(T_i)\ \mbox{and}\ |NC_C(S_v)|\geq 2\}$. Let $\mathcal S=\bigcup_{i=1}^l\mathcal S_i$. Then $\{Y_v=V(S_v)\colon S_v\in\mathcal S\}$ is a desired decomposition of $C^*\setminus (C\cup Y_0)$ (see Fig.\ref{eq14-11-23-2}(f)).

It should be noted that property \eqref{eq14-11-23-2} is used to guarantee that no node is missed in the decomposition. For example, in Fig.\ref{eq14-11-23-2}(c) which does not satisfy this assumption, if we choose $x$ to be the root of the lower subtree, then $S_u=uy$ and $S_y=y$ are stars which does not have at least two terminal neighbors in $T'$, and thus they are excluded from $\mathcal S$. But then, node $y$ does not belong to any star in the decomposition.
\end{proof}

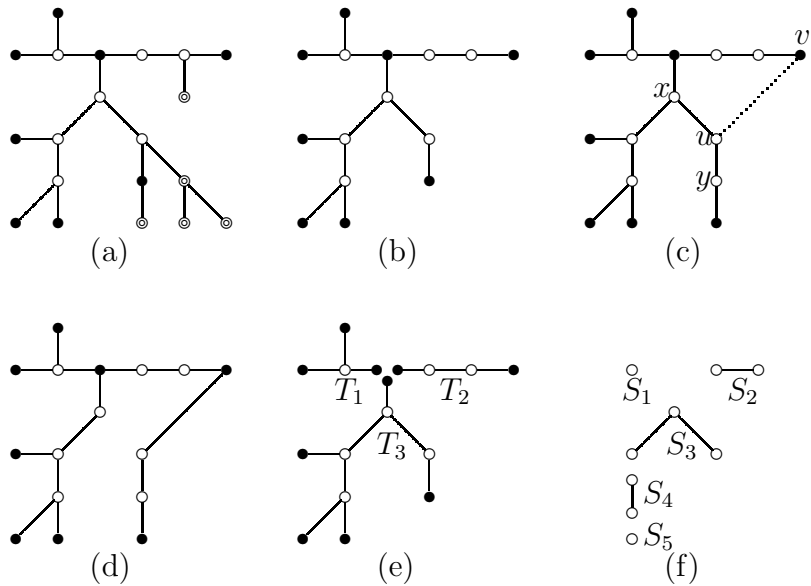
\begin{figure*}[!htbp]\setlength{\unitlength}{0.28mm}
 \begin{center}
 \hskip 0cm\begin{picture}(110,120)
 \put(20,110){\circle*{5}}
 \put(0,90){\circle*{5}}
 \put(20,90){\circle{5}}
 \put(40,90){\circle*{5}}
 \put(60,90){\circle{5}}
 \put(80,90){\circle{5}}
 \put(100,90){\circle*{5}}
 \put(40,70){\circle{5}}
 \put(80,70){\circle{3}}\put(80,70){\circle{5}}
 \put(0,50){\circle*{5}}
 \put(20,50){\circle{5}}
 \put(60,50){\circle{5}}
 \put(20,30){\circle{5}}
 \put(60,30){\circle*{5}}
 \put(80,30){\circle{3}}\put(80,30){\circle{5}}
 \put(0,10){\circle*{5}}
 \put(20,10){\circle*{5}}
 \put(60,10){\circle{3}}\put(60,10){\circle{5}}
 \put(80,10){\circle{3}}\put(80,10){\circle{5}}
 \put(100,10){\circle{3}}\put(100,10){\circle{5}}
 \qbezier(20,107)(20,100)(20,93)
 \qbezier(3,90)(10,90)(17,90)
 \qbezier(23,90)(30,90)(37,90)
 \qbezier(43,90)(50,90)(57,90)
 \qbezier(63,90)(70,90)(77,90)
 \qbezier(83,90)(90,90)(97,90)
 \qbezier(40,87)(40,80)(40,73)
 \qbezier(80,87)(80,80)(80,73)
 \qbezier(38,68)(30,60)(22,52)
 \qbezier(42,68)(50,60)(58,52)
 \qbezier(3,50)(10,50)(17,50)
 \qbezier(20,47)(20,40)(20,33)
 \qbezier(60,47)(60,40)(60,33)
 \qbezier(62,48)(70,40)(78,32)
 \qbezier(18,28)(10,20)(2,12)
 \qbezier(20,27)(20,20)(20,13)
 \qbezier(60,27)(60,20)(60,13)
 \qbezier(80,27)(80,20)(80,13)
 \qbezier(82,28)(90,20)(98,12)
 \put(35,-8){(a)}
 \end{picture}
 \hskip 0.6cm\begin{picture}(110,120)
 \put(20,110){\circle*{5}}
 \put(0,90){\circle*{5}}
 \put(20,90){\circle{5}}
 \put(40,90){\circle*{5}}
 \put(60,90){\circle{5}}
 \put(80,90){\circle{5}}
 \put(100,90){\circle*{5}}
 \put(40,70){\circle{5}}
 \put(0,50){\circle*{5}}
 \put(20,50){\circle{5}}
 \put(60,50){\circle{5}}
 \put(20,30){\circle{5}}
 \put(60,30){\circle*{5}}
 \put(0,10){\circle*{5}}
 \put(20,10){\circle*{5}}
 \qbezier(20,107)(20,100)(20,93)
 \qbezier(3,90)(10,90)(17,90)
 \qbezier(23,90)(30,90)(37,90)
 \qbezier(43,90)(50,90)(57,90)
 \qbezier(63,90)(70,90)(77,90)
 \qbezier(83,90)(90,90)(97,90)
 \qbezier(40,87)(40,80)(40,73)
 \qbezier(38,68)(30,60)(22,52)
 \qbezier(42,68)(50,60)(58,52)
 \qbezier(3,50)(10,50)(17,50)
 \qbezier(20,47)(20,40)(20,33)
 \qbezier(60,47)(60,40)(60,33)
 \qbezier(18,28)(10,20)(2,12)
 \qbezier(20,27)(20,20)(20,13)
 \put(35,-8){(b)}
 \end{picture}
 \hskip 0.6cm\begin{picture}(110,60)
 \put(20,110){\circle*{5}}
 \put(0,90){\circle*{5}}
 \put(20,90){\circle{5}}
 \put(40,90){\circle*{5}}
 \put(60,90){\circle{5}}
 \put(80,90){\circle{5}}
 \put(100,90){\circle*{5}}
 \put(40,70){\circle{5}}
 \put(0,50){\circle*{5}}
 \put(20,50){\circle{5}}
 \put(60,50){\circle{5}}
 \put(20,30){\circle{5}}
 \put(60,30){\circle{5}}
 \put(0,10){\circle*{5}}
 \put(20,10){\circle*{5}}
 \put(60,10){\circle*{5}}
 \qbezier(20,107)(20,100)(20,93)
 \qbezier(3,90)(10,90)(17,90)
 \qbezier(23,90)(30,90)(37,90)
 \qbezier(43,90)(50,90)(57,90)
 \qbezier(63,90)(70,90)(77,90)
 \qbezier(83,90)(90,90)(97,90)
 \qbezier(40,87)(40,80)(40,73)
 {\linethickness{0.25mm}\qbezier[18](98,88)(80,70)(62,52)}
 \qbezier(38,68)(30,60)(22,52)
 \qbezier(42,68)(50,60)(58,52)
 \qbezier(3,50)(10,50)(17,50)
 \qbezier(20,47)(20,40)(20,33)
 \qbezier(60,47)(60,40)(60,33)
 \qbezier(18,28)(10,20)(2,12)
 \qbezier(20,27)(20,20)(20,13)
 \qbezier(60,27)(60,20)(60,13)
 \put(30,70){$x$}\put(50,28){$y$}\put(50,47){$u$}\put(97,95){$v$}
 \put(35,-8){(c)}
 \end{picture}
 \vskip 0.8cm
 \begin{picture}(110,120)
 \put(20,110){\circle*{5}}
 \put(0,90){\circle*{5}}
 \put(20,90){\circle{5}}
 \put(40,90){\circle*{5}}
 \put(60,90){\circle{5}}
 \put(80,90){\circle{5}}
 \put(100,90){\circle*{5}}
 \put(40,70){\circle{5}}
 \put(0,50){\circle*{5}}
 \put(20,50){\circle{5}}
 \put(60,50){\circle{5}}
 \put(20,30){\circle{5}}
 \put(60,30){\circle{5}}
 \put(0,10){\circle*{5}}
 \put(20,10){\circle*{5}}
 \put(60,10){\circle*{5}}
 \qbezier(20,107)(20,100)(20,93)
 \qbezier(3,90)(10,90)(17,90)
 \qbezier(23,90)(30,90)(37,90)
 \qbezier(43,90)(50,90)(57,90)
 \qbezier(63,90)(70,90)(77,90)
 \qbezier(83,90)(90,90)(97,90)
 \qbezier(40,87)(40,80)(40,73)
 \qbezier(98,88)(80,70)(62,52)
 \qbezier(38,68)(30,60)(22,52)
 \qbezier(3,50)(10,50)(17,50)
 \qbezier(20,47)(20,40)(20,33)
 \qbezier(60,47)(60,40)(60,33)
 \qbezier(18,28)(10,20)(2,12)
 \qbezier(20,27)(20,20)(20,13)
 \qbezier(60,27)(60,20)(60,13)
 \put(35,-8){(d)}
 \end{picture}
 \hskip 0.6cm\begin{picture}(110,120)
 \put(20,110){\circle*{5}}
 \put(0,90){\circle*{5}}
 \put(20,90){\circle{5}}
 \put(35,90){\circle*{5}}\put(45,90){\circle*{5}}\put(40,85){\circle*{5}}
 \put(60,90){\circle{5}}
 \put(80,90){\circle{5}}
 \put(100,90){\circle*{5}}
 \put(40,70){\circle{5}}
 \put(0,50){\circle*{5}}
 \put(20,50){\circle{5}}
 \put(60,50){\circle{5}}
 \put(20,30){\circle{5}}
 \put(60,30){\circle*{5}}
 \put(0,10){\circle*{5}}
 \put(20,10){\circle*{5}}
 \qbezier(20,107)(20,100)(20,93)
 \qbezier(3,90)(10,90)(17,90)
 \qbezier(23,90)(30,90)(33,90)
 \qbezier(47,90)(50,90)(57,90)
 \qbezier(63,90)(70,90)(77,90)
 \qbezier(83,90)(90,90)(97,90)
 \qbezier(40,83)(40,80)(40,73)
 \qbezier(38,68)(30,60)(22,52)
 \qbezier(42,68)(50,60)(58,52)
 \qbezier(3,50)(10,50)(17,50)
 \qbezier(20,47)(20,40)(20,33)
 \qbezier(60,47)(60,40)(60,33)
 \qbezier(18,28)(10,20)(2,12)
 \qbezier(20,27)(20,20)(20,13)
 \put(15,76){$T_1$}\put(65,76){$T_2$}\put(35,50){$T_3$}
 \put(35,-8){(e)}
 \end{picture}
 \hskip 0.6cm\begin{picture}(110,60)
 \put(20,90){\circle{5}}
 \put(60,90){\circle{5}}
 \put(80,90){\circle{5}}
 \put(40,70){\circle{5}}
 \put(20,50){\circle{5}}\put(20,38){\circle{5}}
 \put(60,50){\circle{5}}
 \put(20,22){\circle{5}}
 \put(20,10){\circle{5}}
 \qbezier(63,90)(70,90)(77,90)
 \qbezier(38,68)(30,60)(22,52)
 \qbezier(42,68)(50,60)(58,52)
 \qbezier(20,35)(20,30)(20,25)
 \put(15,76){$S_1$}\put(65,76){$S_2$}\put(36,50){$S_3$}\put(25,27){$S_4$}\put(25,7){$S_5$}
 \put(35,-8){(f)}
 \end{picture}
\vskip 0.5cm\caption{ An example for the decomposition of $C^*\setminus C$. Solid nodes are terminal nodes which correspond to components of $G[C]$. In $(a)$, double circle nodes are pruned, yielding tree $T'$ in $(b)$. Figures in $(c)$ and $(d)$ are used to show how to obtain a tree satisfying assumption \eqref{eq14-11-23-2}. In $(c)$, the dashed edge $uv$ is in $G$ but not in $T'$. Adding edge $uv$ and removing edge $ux$ results in the tree in $(d)$. Full components of $T'$ are shown in $(e)$. Figure $(f)$ shows the decomposed stars.}\label{fig11-1-2}
\end{center}
\end{figure*}

\subsection{The Performance Ratio}

\begin{theorem}\label{lem10-10-1}
The connector set $D_2$ of Algorithm \ref{alg10-6-1} has cost $c(D_2)\leq 2H(\delta-1)opt$.
\end{theorem}
\begin{proof}
Let $S_{1},S_{2},\ldots,S_{g}$ be the sets chosen by Algorithm \ref{alg10-6-1} in the order of their selection into set $D_2$. Let $D_G^{(0)}=D_1$. For $1\leq i\leq g$, let
$$
D_G^{(i)}=D_G^{(0)}\cup S_{1}\cup \ldots\cup S_{i}.
$$
For $i=1,\ldots,g$, denote
$$
r_{i}=-\bigtriangleup_{S_{i}} p(D_G^{(i-1)})\ \mbox{and}\ w_{i}=\frac{c(S_{i})}{-\bigtriangleup_{S_{i}}p(D_G^{(i-1)})}.
$$

Suppose $\{Y_0,Y_{1},\ldots,Y_{h}\}$ is the decomposition of $OPT\setminus D_1$ as in Lemma \ref{lem10-6-3}, where $Y_i$ is a star centered at node $v_i$. For $i=1,\ldots,g+1$ and $j=1,\ldots,h$, denote
\begin{equation}\label{eq14-11-26-1}
a_{i,j}=p'_{D_G^{(i-1)}}(Y_j).
\end{equation}
For $1\leq j\leq h$, define
$$
f(Y_{j})=\sum_{i=1}^{g}(a_{i,j}-a_{i+1,j})w_i,
$$
and let
\begin{equation}\label{eq11-25-1}
f(OPT)=\sum_{j=1}^{h}f(Y_{j}).
\end{equation}

{\em Claim 1.} For any $j=1,\ldots,h$, $f(Y_j)\leq H(a_{1,j})c(Y_j)$.

Since $S_i$ is chosen according to Lemma \ref{le10-30-1}, the special structure of $S_i$ implies that
$$
p_{D_G^{(i-1)}}'(S_i)=-\bigtriangleup_{S_i}p(D_G^{(i-1)}).
$$
Hence,
$$
w_{i}=\frac{c(S_{i})}{p_{D_G^{(i-1)}}'(S_i)}.
$$
Then, by the greedy choice of $S_i$, we have
\begin{equation}\label{eq14-11-24-3}
w_i=\frac{c(S_i)}{p_{D_G^{(i-1)}}'(S_i)}\leq \frac{c(Y_{j})}{p_{D_G^{(i-1)}}'(Y_{j})}=\frac{c(Y_j)}{a_{i,j}}.
\end{equation}
By the definition of $p'$, it can be seen that $a_{i,j}$ is a decreasing function on variable $i$. Hence $a_{i,j}-a_{i+1,j}\geq 0$. Combining this with \eqref{eq14-11-24-3},
\begin{align*}
f(Y_{j})\
& \leq \ \sum_{i=1}^g (a_{i,j}-a_{i+1,j})\frac{c(Y_{j})}{a_{i,j}}\\
&\leq\  c(Y_j) \sum_{i=1}^g \big(H(a_{i,j})-H(a_{i+1,j})\big)\\
&=\ c(Y_{j}) \big(H(a_{1,j})-H(a_{g+1,j})\big),
\end{align*}
where the second inequality uses the fact that for any integers $a\geq b$,
$$
\frac{a-b}{a}=\sum_{l=b+1}^a\frac{1}{a}\leq\sum_{l=b+1}^a\frac{1}{l}=H(a)-H(b).
$$
Observe that $a_{g+1,j}=0$ since $D_G^{(g)}$ is connected, the claim follows.

{\em Claim 2.} $c(D_2)\leq f(OPT)$.

Notice that $c(D_2)$ and $f(Y_j)$ can be rewritten as
\begin{equation}\label{eq-11-17-1}
c(D_2)=\sum_{i=1}^gr_iw_i=\sum_{i=1}^g\left(\sum_{l=i}^gr_l-\sum_{l=i+1}^gr_l\right)w_i=
\left(\sum_{l=1}^gr_l\right)w_1+\sum_{i=2}^g\left(\sum_{l=i}^gr_l\right)(w_i-w_{i-1})
\end{equation}
and
\begin{equation}\label{eq-11-17-2}
f(Y_j)=a_{1,j}w_1+\sum_{i=2}^ga_{i,j}(w_i-w_{i-1}).
\end{equation}
By the monotonicity of $p'$, we have $p'_{D_G^{(i-1)}}(S_i)\leq p'_{D_G^{(i-2)}}(S_i)$. Combining this with the greedy choice of $S_{i-1}$, we have
$$
w_i=\frac{c(S_i)}{p'_{D_G^{(i-1)}}(S_i)}\geq \frac{c(S_i)}{p'_{D_G^{(i-2)}}(S_i)}\geq \frac{c(S_{i-1})}{p'_{D_G^{(i-2)}}(S_{i-1})}=w_{i-1}.
$$
In other words, $w_i-w_{i-1}\geq 0$ for $i=1,\ldots,g$. Then, by \eqref{eq11-25-1}, \eqref{eq-11-17-1}, and \eqref{eq-11-17-2}, it can be seen that to prove Claim 2, it suffices to prove that for $i=1,\ldots,g$,
\begin{equation}\label{eq14-11-24-6}
\sum_{j=1}^ha_{i,j}\geq \sum_{l=i}^gr_l.
\end{equation}
The right hand side is
$$
\sum_{l=i}^gr_l=\sum_{l=i}^g-\bigtriangleup_{S_l}p(D_G^{(l-1)})=\sum_{l=i}^g\big(p(D_G^{(l-1)})-p(D_G^{(l)})\big)=p(D_G^{(i-1)})-p(D_G^{(g)})=p(D_G^{(i-1)})-1.
$$
So, proving \eqref{eq14-11-24-6} is equivalent to proving
\begin{equation}\label{eq14-11-24-7}
\sum_{j=1}^hp'_{D^{(i-1)}_G}(Y_j)+1\geq p(D_G^{(i-1)}).
\end{equation}
This inequality can be illustrated by Fig.\ref{figzz}. In this figure, $OPT$ is decomposed into four stars. A comprehension for the value $p'_{D}(Y_1)=3$ is that in Fig.\ref{figzz}(c), the double circled components are merged into the triangled component. Call the new component as $\widetilde C$. Then, the comprehension of $p'_D(Y_2)=3$ is that in Fig.\ref{figzz}(d), double circled components are merged into the triangled component. Notice that this triangled component is contained in $\widetilde C$, and thus we can regard it as $\widetilde C$ in our comprehension. Continue this procedure sequentially in such a way that $Y_1\cup\ldots\cup Y_l$ is connected for $l=1,\ldots,4$. Finally, all components of $G[D]$ are merged into one component, the reduction on the number of components is $p(D)-1$. Notice that the inequality might be strict because some components are counted more than once in the summation part. For example, the component labeled by $u_4$ is repetitively counted.

\vskip 0.2cm 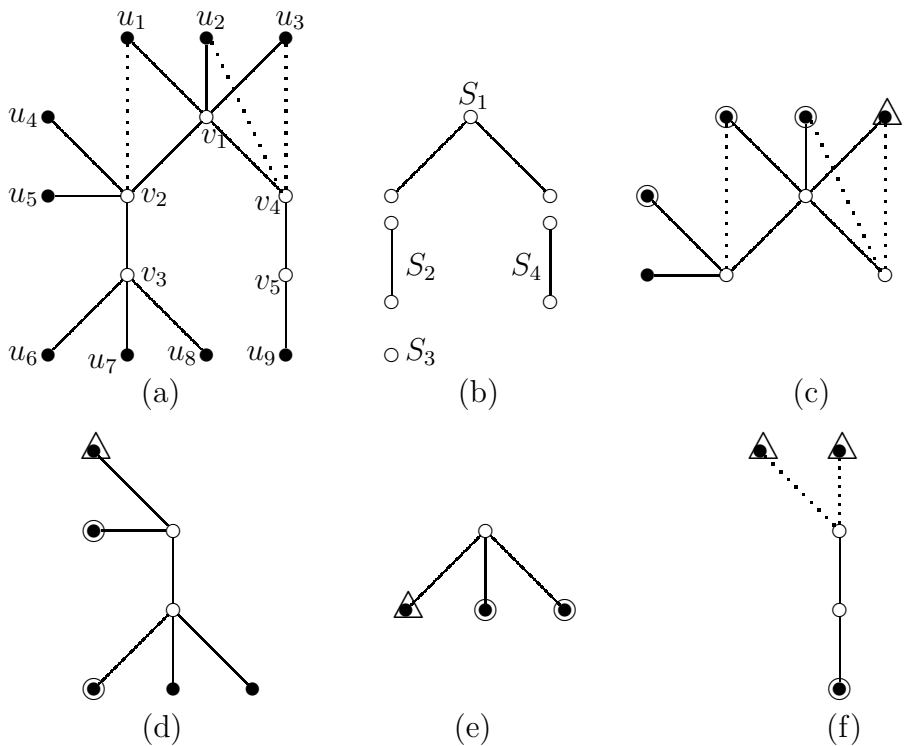
\begin{figure*}[!htbp]
\begin{center}
\begin{picture}(100,140)
\put(0.0,10.0){\circle*{5}}
\put(30.0,10.0){\circle*{5}}
\put(60.0,10.0){\circle*{5}}
\put(90.0,10.0){\circle*{5}}
\put(30.0,40.0){\circle{5}}
\put(90.0,40.0){\circle{5}}
\put(0.0,70.0){\circle*{5}}
\put(30.0,70.0){\circle{5}}
\put(90.0,70.0){\circle{5}}
\put(0.0,100.0){\circle*{5}}
\put(60.0,100.0){\circle{5}}
\put(30.0,130.0){\circle*{5}}
\put(60.0,130.0){\circle*{5}}
\put(90.0,130.0){\circle*{5}}
\qbezier(2,12)(15.0,25.0)(28,38)
\qbezier(30.0,12)(30.0,25.0)(30.0,37)
\qbezier(58,12)(45.0,25.0)(32,38)
\qbezier(90.0,13)(90.0,25.0)(90.0,37)
\qbezier(30.0,43)(30.0,55.0)(30.0,67)
\qbezier(90.0,43)(90.0,55.0)(90.0,67)
\qbezier(3,70.0)(15.0,70.0)(27,70.0)
\qbezier(28,72)(15.0,85.0)(2,98)
\qbezier(32,72)(45.0,85.0)(58,98)
\qbezier(88,72)(75.0,85.0)(62,98)
\qbezier(58,102)(45.0,115.0)(32,128)
\qbezier(60.0,103)(60.0,115.0)(60.0,127)
\qbezier(62,102)(75.0,115.0)(88,128)
{\linethickness{0.3mm}\qbezier[13](30.0,73)(30.0,100.0)(30.0,127)
\qbezier[13](88,72)(75.0,100.0)(62,128)
\qbezier[13](90.0,73)(90.0,100.0)(90.0,127) }
\put(58,90){$v_1$}\put(35,68){$v_2$}\put(35,38){$v_3$}\put(78,65){$v_4$}\put(78,35){$v_5$}
\put(26,135){$u_1$}\put(56,135){$u_2$}\put(86,135){$u_3$}\put(-15,98){$u_4$}\put(-15,68){$u_5$}
\put(-15,8){$u_6$}\put(15,5){$u_7$}\put(45,8){$u_8$}\put(75,8){$u_9$}
\put(35,-8){(a)}
\end{picture}\hskip 0cm\begin{picture}(100,140)
\put(30.0,30.0){\circle{5}}\put(30.0,10.0){\circle{5}}
\put(90.0,30.0){\circle{5}}
\put(30.0,70.0){\circle{5}}\put(30.0,60.0){\circle{5}}
\put(90.0,70.0){\circle{5}}\put(90.0,60.0){\circle{5}}
\put(60.0,100.0){\circle{5}}
\qbezier(30.0,33)(30.0,45.0)(30.0,57)
\qbezier(90.0,33)(90.0,45.0)(90.0,57)
\qbezier(32,72)(45.0,85.0)(58,98)
\qbezier(88,72)(75.0,85.0)(62,98)
\put(55,105){$S_1$}\put(35,40){$S_2$}\put(35,8){$S_3$}\put(75,40){$S_4$}
\put(55,-8){(b)}
\end{picture}
\hskip 0.8cm\begin{picture}(100,140)
\put(0.0,40.0){\circle*{5}}
\put(30.0,40.0){\circle{5}}
\put(90.0,40){\circle{5}}
\put(0.0,70){\circle*{5}}\put(0.0,70){\circle{8}}
\put(60.0,70){\circle{5}}
\put(30.0,100){\circle*{5}}\put(30.0,100){\circle{8}}
\put(60.0,100){\circle*{5}}\put(60.0,100){\circle{8}}
\put(90.0,100){\circle*{5}}{\large\put(84.5,97){$\bigtriangleup$}}
\qbezier(3,40)(15.0,40)(27,40)
\qbezier(28,42)(15.0,55.0)(2,68)
\qbezier(32,42)(45.0,55.0)(58,68)
\qbezier(88,42)(75.0,55)(62,68)
\qbezier(58,72)(45.0,85)(32,98)
\qbezier(60.0,73)(60.0,85)(60.0,97)
\qbezier(62,72)(75.0,85.0)(88,98)
{\linethickness{0.3mm}\qbezier[13](30.0,43)(30.0,70)(30.0,97)
\qbezier[13](88,42)(75.0,70)(62,98)
\qbezier[13](90.0,43)(90.0,70)(90.0,97) }
\put(55,-8){(c)}
\end{picture}

\vskip 0.9cm\begin{picture}(100,100)
\put(30.0,10.0){\circle*{5}}
\put(60.0,10.0){\circle*{5}}
\put(0.0,10.0){\circle*{5}}\put(0.0,10.0){\circle{8}}
\put(30.0,40.0){\circle{5}}
\put(0.0,70.0){\circle*{5}}\put(0.0,70.0){\circle{8}}
\put(30.0,70.0){\circle{5}}
\put(0.0,100){\circle*{5}}{\large\put(-5.5,97){$\bigtriangleup$}}
\qbezier(2,12)(15.0,25.0)(28,38)
\qbezier(30.0,12)(30.0,25.0)(30.0,37)
\qbezier(58,12)(45.0,25.0)(32,38)
\qbezier(30.0,43)(30.0,55.0)(30.0,67)
\qbezier(3,70.0)(15.0,70.0)(27,70.0)
\qbezier(28,72)(15.0,85.0)(2,98)
\put(18,-8){(d)}
\end{picture}
\hskip 0.5cm\begin{picture}(70,100)
\put(30.0,40){\circle*{5}}\put(30,40){\circle{8}}
\put(60.0,40){\circle*{5}}\put(60,40){\circle{8}}
\put(0.0,40){\circle*{5}}{\large\put(-5.5,37){$\bigtriangleup$}}
\put(30.0,70){\circle{5}}
\qbezier(2,42)(15.0,55)(28,68)
\qbezier(30.0,42)(30.0,55)(30.0,67)
\qbezier(58,42)(45.0,55)(32,68)
\put(18,-8){(e)}
\end{picture}
\hskip 0cm\begin{picture}(100,100)
\put(90.0,10.0){\circle*{5}}\put(90.0,10.0){\circle{8}}
\put(90.0,40.0){\circle{5}}
\put(90.0,70.0){\circle{5}}
\put(60.0,100){\circle*{5}}{\large\put(54.8,97){$\bigtriangleup$}}
\put(90.0,100){\circle*{5}}{\large\put(84.6,97){$\bigtriangleup$}}
\qbezier(90.0,13)(90.0,25.0)(90.0,37)
\qbezier(90.0,43)(90.0,55.0)(90.0,67)
{\linethickness{0.3mm}\qbezier[9](88,72)(75,85)(62,98)\qbezier[7](90,73)(90,85)(90,97)}
\put(85,-8){(f)}
\end{picture}\vskip 0.5cm\caption{An illustration for inequality \eqref{eq14-11-24-7}. Solid circles indicate components.}\label{figzz}
\end{center}
\end{figure*}

Claim 2 is proved.

By Lemma \ref{lem10-6-3} $(\romannumeral4)$, we have
\begin{equation}\label{eq14-11-27-8}
\sum_{i=1}^{h}c(Y_{i})\leq 2opt.
\end{equation}
Then by Claim 1, Claim 2, and the observation that $a_{1,j}\leq \delta-1$, we have
$$
c(D_2)\leq \sum_{j=1}^{h}c(Y_j)H(\delta-1)\leq 2H(\delta-1)opt.
$$
The theorem is proved.
\end{proof}

Combining Theorem \ref{thm14-11-25-1} and Theorem \ref{lem10-10-1}, we have the following result.

\begin{theorem}\label{th10-18-2}
Algorithm \ref{alg10-6-1} is a polynomial-time $\big(H(\delta+m)+2H(\delta-1)\big)$-approximation for the $m$-MWCDS problem.
\end{theorem}

\section{Implementation on Unit Disk Graphs}\label{secUDG}

Notice that the maximum degree $\delta$ in the performance ratio of the second part of the algorithm comes from $a_{1,j}\leq \delta-1$. So, to improve the approximation factor when the graph under consideration is a unit disk graph, we improve the upper bound for $a_{1,j}$ first. We use notation $T_1,\ldots,T_l$ in the proof of Lemma \ref{lem10-6-3}. Recall that each $T_i$ has $V(T_i)\subseteq C^*\setminus C$ and  is decomposed into some stars, and the final decomposition of $C^*\setminus C$ is the union of these stars and a set of pruned nodes. We shall use $\|\cdot\|$ to denote Euclidean length.

\begin{lemma}\label{lem14-11-27-7}
In a unit disk graph, there exists a set of subtrees $T_1,\ldots,T_l$ in the proof of Lemma \ref{lem10-6-3} such that any node $v\in V(T_i)$ has $|NC_C(v)|+deg_{T_i}(v)\leq 5$.
\end{lemma}
\begin{proof}
Construct a spanning tree of $G[C\cup C^*]$ as follows. First, replace each component of $G[C]$ by a spanning tree of that component, which is called a {\em tree component}. Then, find a minimum length tree $T$ which spans all nodes of $C^*\setminus C$ and all tree components, where ``minimum length'' is with respect to Euclidean distance. For convenience of statement, we shall call a tree which spans all nodes of $C^*\setminus C$ and all tree components as a {\em valid tree}. Each component of $G[C]$ is called a component node of such a tree and is dealt with as a whole in the following. Similarly to the proof of Lemma \ref{lem10-6-3}, by pruning leafs in $C^*\setminus C$ and by splitting off at component nodes, we obtain a set of full components $T_1',\ldots,T_l'$. Let $T_i$ be the subtree of $T_i'$ with component nodes removed. We choose tree $T$ such that
\begin{equation}\label{eq14-11-27-3}
\sum_{i=1}^l\sum_{u\in V(T_i)}deg_{T_i}(u)\ \mbox{is as small as possible.}
\end{equation}

Suppose there is a node $v\in V(T_i)$ with $|NC_C(v)|+deg_{T_i}(v)=t>6$, assume $NC_C(v)\cup N_{T_i}(v)=\{x_1,\ldots,x_t\}$, where $x_j$ is a node in $C^*\setminus C$ or a node in a component neighbor of $v$ (if $v$ is adjacent with more than one node of a component neighbor, only one node of that component neighbor is chosen to appear in $\{x_1,\ldots,x_t\}$), and $x_1,\ldots,x_t$ are in a clockwise order around node $v$. Since $t>6$, there is an index $j$ with $\angle x_jvx_{j+1}<\pi/3$ ($x_{t+1}$ is viewed as $x_1$). Then $\|x_jx_{j+1}\|<\max\{\|vx_j\|,\|vx_{j+1}\|\}$, say $\|x_jx_{j+1}\|<\|vx_j\|$. Replacing edge $vx_j$ by $x_jx_{j+1}$, we obtain another valid tree whose Euclidean length is shorter than $T$, a contradiction. So,
\begin{equation}\label{eq14-11-27-6}
\mbox{every node $v\in V(T_i)$ has $|NC_C(v)|+deg_{T_i}(v)\leq 6$.}
\end{equation}
A node $u$ with $|NC_C(v)|+deg_{T_i}(v)=6$ is called {\em bad}.

Similar argument shows that for any bad node $v\in V(T_i)$, $\angle x_jvx_{j+1}=\pi/3$ for $j=1,\ldots,6$ and $\|vx_1\|=\cdots=\|vx_6\|$. In other words, $x_1,\ldots,x_6$ locate at the corners of a regular hexagon with center $v$. First, suppose node $v$ has a component neighbor, say $x_1$ is in a component neighbor of $v$. If $x_2\in C$, then $x_2$ must be in a same component of $G[C]$ as $x_1$, because $\|x_1x_2\|=\|vx_1\|\leq 1$, contradicting our convention that one component has at most one node appearing in $\{x_1,\ldots,x_6\}$. So, $x_2\in C^*\setminus C$. Then, $\widetilde T=T-\{vx_2\}+\{x_1x_2\}$ is a valid tree with the same length as $T$. Notice that $\sum_{i=1}^l\sum_{u\in V(T_i)}deg_{T_i}(u)$ is decreased by two (edge $x_1x_2$ does not contribute to the degree sum since one of its end is in a component neighbor and thus does not belong to $T_i$), contradicting the choice of $T$ (see \eqref{eq14-11-27-3}). So, for $j=1,\ldots,6$, $x_j\in C^*\setminus C$. Since $\widetilde T=T-\{vx_{j-1},vx_{j+1}\}+\{x_jx_{j-1},x_jx_{j+1}\}$ is also a valid tree whose length is the same as $T$, we have $|NC_C(x_j)|+deg_{\widetilde T_i}(x_j)\leq 6$. By noticing that $deg_{\widetilde T_i}(x_j)=deg_{T_i}(x_j)+2$ (since both $x_{j-1},x_{j+1}\in C^*\setminus C$), we have
\begin{equation}\label{eq14-11-27-5}
|NC_C(x_j)|+deg_{T_i}(x_j)\leq 4.
\end{equation}
This inequality holds for any $j=1,\ldots,6$. Notice that $\widehat T=T-\{vx_1\}+\{x_1x_2\}$ is a valid tree with the same length as $T$ and $\sum_{i=1}^l\sum_{u\in V(\widehat T_i)}deg_{\widehat T_i}(u)=\sum_{i=1}^l\sum_{u\in V(T_i)}deg_{T_i}(u)$. By property \eqref{eq14-11-27-6} and \eqref{eq14-11-27-5}, we have
$$
|NC_C(u)|+deg_{\widehat T_i}(u)=\left\{\begin{array}{ll}|NC_C(u)|+deg_{T_i}(u)-1\leq 5, & u=v,\\|NC_C(u)|+deg_{T_i}(u)+1\leq 5, & u=x_1\ \mbox{or}\ x_2,\\ |NC_C(u)|+deg_{T_i}(u)\leq 6, & u\neq v,x_1,x_2.\end{array}\right.
$$
So, the number of bad nodes in $\widehat T$ is strictly reduced. By recursively executing such an operation, we have a tree satisfying the requirement of this lemma.

Recall that condition \eqref{eq14-11-23-2} plays an important role in the decomposition. This does no pose any difficulty here, because by the modification method in the proof of Lemma \ref{lem10-6-3}, if $u\in C^*\setminus C$ is not adjacent with any component neighbor in $T'$, then we may just add an edge between $u$ and a component neighbor, and remove an edge between $u$ and another node in $C^*\setminus C$. Such an operation does not increase the value of $|NC_C(u)|+deg_{T_i}(u)$.
\end{proof}

\begin{theorem}
When applied to unit disk graphs, the node set $D_2$ produced by Algorithm \ref{alg10-6-1} has cost $c(D_2)\leq 3.67opt$.
\end{theorem}
\begin{proof}
Let $T_1,\ldots,T_l$ be the subtrees in Lemma \ref{lem14-11-27-7}. Choose node $v_i=\arg\max\{c(v)\colon v\in V(T_i)\}$ to be the root of $T_i$. Decompose $C^*\setminus C$ as in the proof of Lemma \ref{lem10-6-3}. Let $\{Y_v^{(i)}\}$ be the set of stars coming from the decomposition of $T_i$. To avoid ambiguity, we use $a_{1,Y_v^{(i)}}$ to denote $p'_{D_1}(Y_v^{(i)})$ (which is $a_{1,j}$ in the proof of Theorem \ref{lem10-10-1}). By Lemma \ref{lem14-11-27-7}, if $v=v_i$, then $a_{1,Y_v^{(i)}}\leq |NC_{D_1}(v)|-1+ deg_{T_i}(v)\leq 4$; if $v\neq v_i$, then  $a_{1,Y_v^{(i)}}\leq |NC_{D_1}(v)|-1+deg_{T_i}(v)-1\leq 3$ (this is because the parent of $v$ is not in $Y_v$ if $v\neq v_i$).

Notice that $v_i$ belongs to exactly one star in the decomposition of $C^*\setminus C$. Hence inequality \eqref{eq14-11-27-8} can be improved to
\begin{equation}\label{eq14-11-27-9}
\sum_{i=1}^l\sum_{Y_v^{(i)}}c(Y_v^{(i)})+\sum_{i=1}^lc(v_i)\leq 2opt.
\end{equation}
Combining Lemma \ref{lem14-11-27-7} with Claim 1 and Claim 2 of Theorem \ref{lem10-10-1},
\begin{equation}\label{eq14-11-27-10}
c(D_2)\leq \sum_{i=1}^l\sum_{Y_v^{(i)}}c(Y_v^{(i)})H(a_{1,Y_v^{(i)}})\leq \sum_{i=1}^l\left(H(4)c(Y_{v_i}^{(i)})+H(3)\sum_{Y_v^{(i)},v\neq v_i}c(Y_v^{(i)})\right).
\end{equation}
Since $D_1$ is an $m$-fold dominating set, every $v_i$ has at least one component neighbor. Then by Lemma \ref{lem14-11-27-7}, $deg_{T_i}(v_i)\leq 4$, and thus $Y_{v_i}^{(i)}$ has at most five nodes. Since $v_i$ has the maximum cost in $T_i$, we have $c(Y_{v_i}^{(i)})\leq 5c(v_i)$. So,
$$
H(4)c(Y_{v_i}^{(i)})=H(3)c(Y_{v_i}^{(i)})+\frac{c(Y^{(i)}_{v_i})}{4}\leq H(3)c(Y_{v_i}^{(i)})+5c(v_i)/4<H(3)\big(c(Y_{v_i}^{(i)})+c(v_i)\big).
$$
Combining this inequality with \eqref{eq14-11-27-9} and \eqref{eq14-11-27-10},
$$
c(D_2)\leq H(3)\left(\sum_{i=1}^l\sum_{Y_v^{(i)}}c(Y_v^{(i)})+\sum_{i=1}^lc(v_i)\right)\leq 2H(3)opt<3.67opt.
$$
The theorem is proved.
\end{proof}

\section{Conclusion}\label{sec.7}

In this paper, we presented a $(H(\delta+m)+2H(\delta-1))$-approximation algorithm for the minimum weight $(1,m)$-CDS problem, where $\delta$ is the maximum degree of the graph. Compared with the $1.35\ln n$-approximation algorithm for the minimum $(1,1)$-CDS problem \cite{Guha1}, our constant is larger. However, since in many cases, the maximum degree is much smaller than the number of nodes, our result is an improvement on the performance ratio in some sense. In particular, the replacement of $n$ by $\delta$ in the performance ratio makes it possible to obtain a 3.67-approximation for the connecting part when the topology of the network is a unit disk graph. In fact, Zou {\it et al.} obtained a $2.5\rho$-approximation for the connecting part in a unit disk graph, where $\rho$ is the performance ratio for the minimum Steiner tree problem. If the best $\rho=1.39$ is used, their algorithm has performance ratio $3.475$. Notice that the $1.39$-approximation algorithm for the minimum Steiner tree problem uses a combination of iterated rounding and random rounding. Our algorithm has the advantage of purely combinatorial and deterministic. Furthermore, we expect our method to have a theoretical value which can be used to deal with other problems.

\section*{Acknowledgements}
This research is supported by NSFC (61222201,11531011), SRFDP
(20126501110001), and Xingjiang Talent Youth Project (2013711011). It is accomplished when the second author is visiting National Chiao Tung University, Taiwan, sponsored by ``Aiming for the Top University Program'' of the National Chiao Tung University and Ministry of Education, Taiwan.

\end{document}